\documentclass[conference]{IEEEtran}

\ifCLASSINFOpdf
\else
\fi

\usepackage{amsmath,amssymb,amsfonts}
\usepackage{cite}

\usepackage{algorithmic}
\usepackage{stfloats}
\usepackage{graphicx}
\usepackage{textcomp}
\usepackage{xcolor}
\usepackage{pifont}
\usepackage{amsthm}

\usepackage{mathrsfs}
\usepackage{cases}
\usepackage{comment}
\usepackage{empheq}
\usepackage{caption}
\usepackage{subcaption}
\usepackage{bm}

\allowdisplaybreaks

\usepackage{url}  
\usepackage{graphicx}  
\usepackage[ruled,vlined,linesnumbered]{algorithm2e}
\usepackage{setspace}
\usepackage{floatpag}
\usepackage{float}
\floatpagestyle{empty}

\newtheorem{thm}{Theorem}

\usepackage{listings}

\usepackage{fancyhdr}

\def\BibTeX{{\rm B\kern-.05em{\sc i\kern-.025em b}\kern-.08em
    T\kern-.1667em\lower.7ex\hbox{E}\kern-.125emX}}
\begin{document}
\captionsetup[figure]{labelfont={rm},labelformat={default},labelsep=period,name={Fig.}}

\title{Secrecy Rate Maximization for Intelligent Reflecting Surface Aided SWIPT Systems}
\author
{\IEEEauthorblockN{Wei Sun\textsuperscript{1,2,3},~Qingyang Song*\textsuperscript{2},~Lei Guo\textsuperscript{2},~Jun Zhao\textsuperscript{3}}\\ [-10pt]
\IEEEauthorblockA{\textsuperscript{1}School of Computer Science and Engineering, Northeastern University, Shenyang, China\\
\textsuperscript{2}School of Communication and Information Engineering, Chongqing University of Posts\\ and Telecommunications,
Chongqing, China\\
\textsuperscript{3}School of Computer Science and Engineering,
Nanyang Technological University, Singapore\\
Emails:~weisun@stumail.neu.edu.cn, \{songqy,guolei\}@cqupt.edu.cn, JunZhao@ntu.edu.sg.\\[-13pt] }
}

\renewcommand{\headrulewidth}{0pt}

\maketitle
\thispagestyle{fancy}
\pagestyle{fancy}
\lhead{This paper appears in the Proceedings of 2020 IEEE/CIC International Conference on Communications in China (\textbf{ICCC}) 2020. Please feel free to contact us for questions or remarks.}
\cfoot{\thepage}
\renewcommand{\headrulewidth}{0.4pt}
\renewcommand{\footrulewidth}{0pt}

\begin{abstract}
Simultaneous wireless information and power transfer (SWIPT) and intelligent reflecting surface (IRS) are two promising techniques for providing enhanced wireless communication capability and sustainable energy supply to energy-constrained wireless devices. Moreover, the combination of the IRS and the SWIPT can create the ``one plus one greater than two" effect. However, due to the broadcast nature of wireless media, the IRS-aided SWIPT systems are vulnerable to eavesdropping. In this paper, we study the security issue of the IRS-aided SWIPT systems. The objective is to maximize the secrecy rate by jointly designing the transmit beamforming and artificial noise (AN) covariance matrix at a base station (BS) and reflective beamforming at an IRS, under transmit power constraint at the BS and energy harvesting (EH) constraints at multiple energy receivers. To tackle the formulated non-convex problem, we first employ an alternating optimization (AO) algorithm to decouple the coupling variables. Then, reflective beamforming, transmit beamforming and AN covariance matrix can be optimized by using a penalty-based algorithm and semidefinite relaxation (SDR) method, respectively. Simulation results demonstrate the effectiveness of the proposed scheme over baseline schemes.
\end{abstract}

\begin{IEEEkeywords}
Secrecy rate maximization, IRS, SWIPT, SDR, penalty-based algorithm.
\end{IEEEkeywords}


\section{Introduction}
With the commercial deployments of fifth-generation (5G), the exploration of next-generation (i.e., 6G) communication technologies has been begun in both academia and industry \cite{8869705,8808168,feng2019cooperative, 9086448,du2020mec}. Compared with previous generations, the goal of 6G is not only to pursue the improvement of network capacity and transmission rate but also to achieve intelligent interconnection of everything. With the large-scale deployment and connection of battery-constrained devices in 6G, it is urgent to continuously supply power to devices to prolong the lifetime of networks. In recent years, the emergence of simultaneous wireless information and power transfer (SWIPT) technology is expected to solve this problem \cite{qi2020robust,alageli2018swipt,8636993}. It can reduce the device's dependence on batteries, and provide enough energy to support higher-performance communication.\par

Nevertheless, the randomness of wireless channels and severe channel attenuation will lead to weak energy and information received at receivers in SWIPT systems. Recently, intelligent reflecting surface (IRS) as an intelligent technology is proposed to overcome the harmful effects of wireless environments \cite{wu2019towards, wu2019beamforming, wu2019joint, pan2019intelligent, tang2019joint}. The IRS consists of many low-cost passive reflection units, and each of them can be controlled by software to change the phase and amplitude of incident signals. By introducing IRS into the SWIPT systems and properly adjusting the reflective beamforming at the IRS, the reflected signal can be superimposed with the signals from other paths to enhance received signal power. It can also eliminate interference signals and effectively improve energy transmission efficiency. Some works have introduced IRS into the SWIPT systems to improve system performance. In \cite{wu2019joint}, the authors studied the transmit power minimization problem by jointly designing active and passive beamforming in the IRS-aided SWIPT systems. Pan \textit{et al}.~extended the case to the IRS-aided multi-input multi-output (MIMO) SWIPT system to maximize the weighted sum-rate \cite{pan2019intelligent}. Tang \textit{et al}.~jointly optimized information and energy transmit beamforming at the base station (BS) and passive beamforming at the IRS to maximize the minimum received power at all energy harvesting (EH) receivers \cite{tang2019joint}. \par

However, due to the broadcast nature of wireless media, the IRS-aided SWIPT systems are vulnerable to eavesdropping. Therefore, secure communication is a particularly important issue. Shen \textit{et al}.~investigated transmission optimization problem for IRS-assisted multi-antenna systems from a physical layer security (PLS) perspective \cite{8743496}. In \cite{yu2019robust}, a robust transmit beamforming algorithm has been designed to achieve secure communication in an IRS-aided system. In \cite{feng2019physical}, the authors proposed two algorithms to maximize the secrecy rate in an IRS-aided multiuser multi-input single-output (MISO) wireless system. In \cite {8421218}, Tang \textit{et al}.~studied secure non-orthogonal multiple access (NOMA) SWIPT systems. The sum secrecy rate is maximized under constraints on the minimum data rate requirement and the minimum harvested energy requirement. It can be seen that the PLS problem has been widely studied in the IRS-aided or SWIPT-aided wireless communication systems. However, it has not been fully explored in the IRS-aided SWIPT systems.

In this paper, we investigate PLS provisioning for IRS-aided SWIPT systems, where a multiple-antenna BS sends a message to an information receiver (IR) and provides energy to energy receivers (ERs) simultaneously in the presence of eavesdroppers (Eves). Meanwhile, ERs are also regarded as potential Eves. To further improve PLS, artificial noise (AN) is considered at the BS. Assuming that the channel state information (CSI) of the Eves is available, the secrecy rate is maximized by jointly optimizing the transmit beamforming and AN covariance matrix at the BS and the reflective beamforming at the IRS, under BS transmit power constraint and EH constraints. To this end, the secure beamforming design is formulated as a non-convex optimization problem. The non-convexities of the EH constraints and the unit modulus constraints at the IRS are main challenges to solve this problem. First, we propose an effective alternating optimization (AO) algorithm to address this non-convex problem. Next, the design of transmit beamforming and AN covariance matrix is resolved by using the semidefinite relaxation (SDR) method. The reflective beamforming is solved by using a penalty-based algorithm. Simulations show that our proposed scheme is better than two baselines in terms of secrecy rate.

The rest of this paper is organized as follows. The system model of the proposed secure IRS-aided SWIPT system and problem formulation are presented in Section \ref{II}. In Section \ref{III}, an effective iterative algorithm for maximizing the secrecy rate is proposed. Simulation results and discussions are presented in Section \ref{IV}. Finally, this paper is concluded in Section \ref{V}.\par

\textit{Notations:} Vectors and matrices are denoted by boldface lower-case and capital letters, respectively. $\bm A ^{H}$ and $\|{\bm A}\|$ denote Hermitian operator and Euclidean norm of ${\bm A}$, respectively. The symbols $\mathrm{Tr}(\cdot)$ and $\mathbb{E}[\cdot]$ denote the trace and statistical expectation, respectively. $\mathrm{diag}\{\bm A\}$ is a diagonal matrix with the entries of $\bm A$ on its main diagonal. ${\bm A}\succeq \boldsymbol{0}$ indicates that $\bm A$ is a positive semi-definite (PSD) matrix. $\mathbb{C}^{x \times y}$ denotes the space of $x \times y$ complex-valued matrix. The circularly symmetric complex Gaussian (CSCG) distribution is denoted by $\mathcal{CN}(0,\sigma^2)$ with mean $0$ and variance $\sigma^2$. \par

\section{System Model And Problem Formulation}\label{II}
\subsection{System Model}
As illustrated in Fig. $\ref{fig:1}$, we consider a secure IRS-aided SWIPT system that consists of one $N_{t}$-antenna BS, one IRS, one IR, $K$ Eves, indexed by $\mathcal{K}\triangleq \{1,\cdots,K\}$, and $M$ ERs, indexed by $\mathcal{M}\triangleq \{1,\cdots,M\}$. All receivers and Eves are equipped with a single antenna. The IRS consists of $N_{r}$ passive reflecting elements, indexed by $\mathcal{N}_r \triangleq \{1,\cdots,N_r\}$. The IRS assists secure SWIPT transmission from the BS to the IR and the ERs. The ERs are usually deployed close to the BS to achieve a higher EH efficiency. However, it makes them easier to eavesdrop on the IR signal. Therefore, they are also regarded as potential Eves. Due to the broadcast nature of radio frequency (RF) channels, the signal transmitted from the BS to the IR is overheard by the ERs and the Eves. \par
As shown in Fig. $\ref{fig:1}$, the baseband equivalent channel responses from the BS to the IRS, from the BS to the IR, from the BS to the $i$th ER, from the BS to the $k$th Eve, from the IRS to the IR, from the IRS to the $k$th Eve, and from the IRS to the $i$th ER are denoted by $\boldsymbol{Q}\in\mathbb{C}^{N_{r}\times N_{t}}$, $\boldsymbol{h}_{d}\in\mathbb{C}^{N_{t}\times 1}$, $\boldsymbol{g}_{d,i}\in\mathbb{C}^{N_{t}\times 1}$, $\boldsymbol{h}_{deve,k}\in\mathbb{C}^{N_{t}\times 1}$, $\boldsymbol{h}_{r}\in\mathbb{C}^{N_{r}\times 1}$, $\boldsymbol{h}_{reve,k}\in\mathbb{C}^{N_{r}\times 1}$ and $\boldsymbol{g}_{r,i}\in\mathbb{C}^{N_{r}\times 1}$, respectively. All channels are assumed to be quasi-static flat fading.\par
\begin{figure}[!t]
\centering
\includegraphics[scale=0.4]{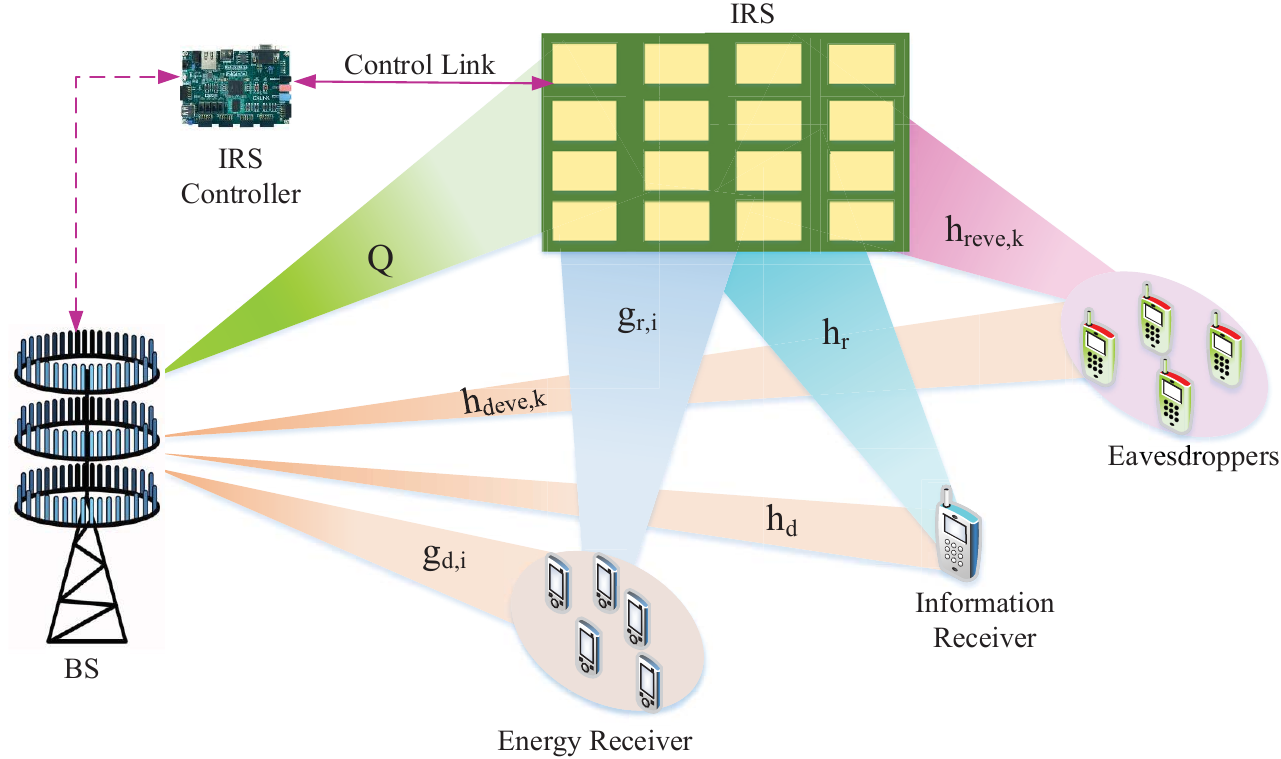}
\caption{System model for secure IRS-aided SWIPT systems.}
\label{fig:1}
\vspace{-14pt}
\end{figure}
To protect the data transmission against eavesdropping and improve the amount of harvested energy at the ERs, the BS transmits the information signal together with AN signal to the IR. Thus, the transmit signal $\boldsymbol{x}$ from the BS is given by $\boldsymbol{x}=\boldsymbol{w}s+\boldsymbol{v}$, where $s\in\mathbb{C}$ is the transmit symbol with $\mathbb{E}[ss^{H}]=1$, $\boldsymbol{w}\in\mathbb{C}^{N_{t}\times 1}$ indicates the transmit beamforming vector which sends the confidential information to the desired receiver and $\boldsymbol{v} \in \mathbb{C}^{N_{t}\times 1}$ is a pseudo-random AN vector generated by the BS. It is assumed that $\boldsymbol{v}$ is modeled as a random vector with CSCG distribution, i.e., $\boldsymbol{v}\sim \mathcal{CN}(0,\boldsymbol{V})$ with $\boldsymbol{V} \succeq 0$.\par
The received signal at the IR is written as
\vspace{-3pt}
\begin{align}
y_{ir}=\boldsymbol{h}^H\boldsymbol{x}+n_{ir},\label{2}
\end{align}
where $\boldsymbol{h}^H=\boldsymbol{h}_{d}^{H}+\boldsymbol{h}_{r}^{H}\boldsymbol{\Theta}\boldsymbol{Q}$, $\boldsymbol{\Theta}=\mathrm{diag}(\theta_1,\cdots, \theta_{N_r})\in\mathbb{C}^{N_{r}\times N_{r}}$ is the reflective matrix at the IRS, $\theta_n=e^{j\phi_{n}}$, $\phi_{n} \in [0, 2\pi)$ is the phase shift, and $n_{ir} \thicksim \mathcal{CN}(0, \sigma^{2}_{ir})$ is the complex additive white Gaussian noise (AWGN) at the IR.

The achievable rate of the IR is given by
\vspace{-3pt}
\begin{align}
R_{ir}=\log_{2}\left(1+|\boldsymbol{h}^{H}\boldsymbol{w}|^{2}/\sigma_{ir}^{2}\right),\label{6}
\end{align}
It is noted that the AN signal is a known deterministic sequence at the IR. Then, interference cancellation techniques can be used to cancel the AN signal before decoding the desired signal, thus it has no impact on the achievable rate \cite{su2018physical}.\par

Similarly, the received signals at the $i$th ER and the $k$th Eve are given by
\vspace{-3pt}
\begin{align}
y_{e,i}&=\boldsymbol{g}^{H}_{i}\boldsymbol{x}+n_{e,i},\\
y_{eve,k}&=\boldsymbol{h}^H_{eve,k}\boldsymbol{x}+n_{eve,k}\label{3},
\end{align}
where $\boldsymbol{g}^{H}_{i}=\boldsymbol{g}^{H}_{d,i}+\boldsymbol{g}^{H}_{r,i}\boldsymbol{\Theta}\boldsymbol{Q}$, $\boldsymbol{h}^H_{eve,k}=\boldsymbol{h}^{H}_{deve,k}+\boldsymbol{h}^{H}_{reve,k}\boldsymbol{\Theta}\boldsymbol{Q}$, $n_{eve,k}$ and $n_{e,i}$ are the complex AWGN at the $k$th Eve and the $i$th ER respectively, each of which is with zero mean and variances $\sigma^{2}_{eve,k}$ and $\sigma^2_{e,i}$.

The eavesdropping rates at the $i$th ER and the $k$th Eve can be expressed as
\vspace{-3pt}
\begin{align}
R_{e,i}&=\log_{2}\left(1+\frac{|\boldsymbol{g}^{H}_{i}\boldsymbol{w}|^2}{\mathrm{Tr}(\boldsymbol{g}^{H}_{i}\boldsymbol{g}_{i}\boldsymbol{V})+\sigma_{e,i}^{2}}\right),\label{6-2}\\
R_{eve,k}&=\log_{2}\left(1+\frac{|{\boldsymbol{h}}^{H}_{eve,k}\boldsymbol{w}|^2}{\mathrm{Tr}({\boldsymbol{h}}^{H}_{eve,k}{\boldsymbol{h}}_{eve,k}\boldsymbol{V})+\sigma_{eve,k}^{2}}\right).\label{7}
\end{align}

Thus, the achievable secrecy rate can be written as \cite{yu2019robust}
\begin{equation}
R_{sec}=[R_{ir}-\max\{\max_{k \in \mathcal{K}} \{R_{eve,k}\},\max_{i \in \mathcal{M}} \{R_{e,i}\}\}]^{+},\label{9}
\end{equation}
where $[a]^+=\max(0,a)$.

The received power at the $i$th ER can be given by \cite{tang2019joint}
\begin{eqnarray}
P_{EH,i}=|\boldsymbol{g}^{H}_{i}\boldsymbol{w}|^2+\mathrm{Tr}(\boldsymbol{g}^{H}_{i}\boldsymbol{g}_{i}\boldsymbol{V}).
\end{eqnarray}
In this paper, we employ a practial non-linear energy harvesting model proposed in \cite{boshkovska2017robust,boshkovska2015practical}. Then, the total harvested power at the $i$th ER is modeled as
\begin{eqnarray}\label{4}
\Phi_{EH,i}=\frac{\Psi_{EH,i}-M_i \Omega_i}{1-\Omega_i},~\Omega_i=\frac{1}{1+\mathrm{exp}(a_i b_i)},\\
\Psi_{EH,i}=\frac{M_i}{1+\mathrm{exp}(-a_i(P_{EH,i}-b_i))},
\end{eqnarray}
where $M_i$ is the maximum harvested power at the $i$th ER when the energy harvesting circuit is saturated, $a_i$ and $b_i$ are constants which capture the joint effects of resistance, capacitance and circuit sensitivity \cite{boshkovska2015practical}. Since the noise power is much smaller than the RF signal power, it can be ignored \cite{boshkovska2017robust,boshkovska2015practical}.

\subsection{Problem Formulation}
In this paper, the joint optimization of transmit beamforming, AN covariance matrix at the BS and reflective beamforming at the IRS is studied to maximize the secrecy rate in the secure IRS-aided SWIPT system, subject to transmit power constraint at the BS, EH constraints at the ERs as well as the constraint on the IRS reflective beamforming. The considered optimization problem is mathematically formulated as
\begin{subequations}
\begin{align}
\mathcal{P}_1:
\max_{\boldsymbol{w},\boldsymbol{V},\boldsymbol{\Theta}}~&R_{sec} &\label{10a}\\
\mbox{s.t.}~
&\|\boldsymbol{w}\|^2+\mathrm{Tr}(\boldsymbol{V}) \leq P_s,&\label{10b}\\
&\Phi_{EH,i}\geq \mu_{i},~\forall i \in \mathcal{M},&\label{10c}\\
&|\theta_{n}|=1,~\forall n \in \mathcal{N}_r,&\label{10d}\\
&\boldsymbol{V} \succeq \boldsymbol{0},&\label{10e}
\end{align}
\end{subequations}
where $P_s$ is the maximum transmit power at the BS, and $\mu_{i}$ denotes the minimum harvested power requirement for the $i$th ER.

It is obvious that $\mathcal{P}_1$ is non-convex nonlinear programming due to the non-convex objective function and constraints. So $\mathcal{P}_1$ is difficult to solve directly. In the following section, after some transformations, two algorithms will be proposed to deal with $\mathcal{P}_1$ efficiently in an iterative manner.
\section{Algorithm Design for Secure IRS-aided SWIPT Systems}\label{III}
In this section, we first reformulate the original problem into a more tractable form. Next, the SDR method and penalty-based algorithm are proposed to solve $\mathcal{P}_1$ in an alternative manner.

\subsection{Problem Transformation}
Firstly, we recast (\ref{10c}) as follows.
\begin{align}
P_{EH,i} \geq \beta_i:=b_i-\frac{1}{a_i}\mathrm{ln}(\frac{M_i}{\mu_i(1-\Omega_i)+M_i\Omega_i}-1),
\end{align}
where $\beta_i$ denotes the required received power under the non-linear EH model.\par
To solve problem $\mathcal{P}_1$, we introduce an auxiliary variable $\tau$ as the maximum tolerable channel capacity of each eavesdropper (i.e., ERs and Eves). Then, we can get
\begin{align}\label{transform-eu-eve}
R_{e,i} \leq \tau,~\forall i \in \mathcal{M}, ~R_{eve,k}\leq \tau,~\forall k \in \mathcal{K}.
\end{align}

Then, we recast $\mathcal{P}_1$ in an equivalent form as follows.
\begin{subequations}
\begin{align}
\mathcal{P}_2:
\max_{\boldsymbol{w},\boldsymbol{V},\boldsymbol{\Theta},\tau}&~ [R_{ir}-\tau]^{+} \label{14a}\\
\mbox{s.t.}~~
&\mathrm{Tr}\left(\boldsymbol{w}\boldsymbol{w}^{H}\right)+\mathrm{Tr}\left(\boldsymbol{V}\right)\leq P_{s},&\label{14d}\\
&P_{EH,i} \geq \beta_{i},&\label{14e}\\
&(\text{\ref{10d}}),~(\text{\ref{10e}}),~(\text{\ref{transform-eu-eve}}).
\end{align}
\end{subequations}

The problem $\mathcal{P}_2$ is non-convex due to constraints (\ref{transform-eu-eve}) and (\ref{10d}). Inspired by \cite{yu2019robust}, to deal with the non-convexity of constraint (\ref{transform-eu-eve}), for given $\tau$, we jointly optimize the transmit beamforming and AN covariance matrix at the BS and reflective beamforming at the IRS to maximize the secrecy rate. For a fixed $\tau$, set $\gamma=2^{\tau}-1$. Then, the optimization problem $\mathcal{P}_2$ can be transformed into the following equivalent form.
\begin{subequations}
\begin{align}
\mathcal{P}_3:
\min_{\boldsymbol{w},\boldsymbol{V},\boldsymbol{\Theta}}&~ -R_{ir} \label{P3_a}\\
\mbox{s.t.}~~
&\mathrm{Tr}\left(\boldsymbol{w}\boldsymbol{w}^{H}\right)+\mathrm{Tr}\left(\boldsymbol{V}\right)\leq P_{s},&\label{P3_b}\\
&P_{EH,i} \geq \beta_{i},&\label{P3_c}\\
&\frac{|\boldsymbol{g}^{H}_{i}\boldsymbol{w}|^2}{\mathrm{Tr}(\boldsymbol{g}^{H}_{i}\boldsymbol{g}_{i}\boldsymbol{V})+\sigma_{e,i}^{2}} \leq \gamma,&\label{P3_d}\\
&\frac{|{\boldsymbol{h}}^H_{eve,k}\boldsymbol{w}|^2}{\mathrm{Tr}({\boldsymbol{h}}^H_{eve,k}{\boldsymbol{h}}_{eve,k}\boldsymbol{V})+\sigma_{eve,k}^{2}} \leq \gamma,&\label{P3_e}\\
&(\text{\ref{10d}}),~(\text{\ref{10e}}).
\end{align}
\end{subequations}

It is shown that problem $\mathcal{P}_3$ is still non-convex due to tightly coupled transmit beamforming, AN covariance matrix and reflective beamforming. We employ an AO algorithm to decouple the optimization variables. Specifically, $\{\boldsymbol{w}, \boldsymbol{V}\}$ and $\boldsymbol{\Theta}$ are alternately solved while fixing the other variables. Then, the original problem can be divided into two subproblems. Two effective algorithms are proposed to solve $\mathcal{P}_3$ in an iterative manner.

\subsection{Transmit Beamforming and AN Covariance Matrix Design at the BS}\label{subB}
For a given phase shift matrix $\boldsymbol{\Theta}$, we optimize the beamforming vector $\boldsymbol{w}$ and AN covariance matrix $\boldsymbol{V}$. In this way, the optimization problem $\mathcal{P}_3$ can be converted into $\mathcal{P}_4$ under the following definitions: $\boldsymbol{G}_i=\boldsymbol{g}_i\boldsymbol{g}_i^H$, $\boldsymbol{H}=\boldsymbol{hh}^H$, $\boldsymbol{W}=\boldsymbol{w}\boldsymbol{w}^H$, ${\boldsymbol{H}}_{eve,k}={\boldsymbol{h}}_{eve,k}{\boldsymbol{h}}_{eve,k}^H$. \par
\vspace{-10pt}
\begin{subequations}
\begin{align}
\mathcal{P}_4:
\min_{\boldsymbol{W},\boldsymbol{V}}&~ -\log_2(1+\alpha/\sigma^2_{ir})\\
\mbox{s.t.}~~
&\mathrm{Tr}(\boldsymbol{W})+\mathrm{Tr}(\boldsymbol{V})\leq P_{s},&\label{P4_b}\\
&\mathrm{Tr}(\boldsymbol{G}_{i}\boldsymbol{W})+\mathrm{Tr}(\boldsymbol{G}_{i}\boldsymbol{V}) \geq \beta_{i},&\label{P4_c}\\
&\mathrm{Tr}(\boldsymbol{G}_{i}\boldsymbol{W})-\gamma\mathrm{Tr}(\boldsymbol{G}_{i}\boldsymbol{V})\leq \gamma \sigma_{e,i}^{2},\label{P4_d}\\
&\mathrm{Tr}({\boldsymbol{H}}_{eve,k}\boldsymbol{W})-\gamma \mathrm{Tr}({\boldsymbol{H}}_{eve,k}\boldsymbol{V})\leq \gamma \sigma_{eve,k}^{2},\label{P4_e}\\
&\mathrm{Tr}(\boldsymbol{WH})\geq \alpha,~\boldsymbol{W} \succeq \boldsymbol{0}, \boldsymbol{V} \succeq \boldsymbol{0},\label{P4_f}\\
&\mathrm{Rank}(\boldsymbol{W})\le 1.\label{P4_g}
\end{align}
\end{subequations}\par
After the conversion from $\mathcal{P}_3$ to $\mathcal{P}_4$, the only non-convexity in $\mathcal{P}_4$ is the rank constraint in (\ref{P4_g}). We utilize the SDR method to solve this problem. Specifically, $\mathcal{P}_4$ can be converted to a standard convex optimization problem after removing constraint (\ref{P4_g}), and then it can be solved by CVX toolbox\cite{boyd2004convex}. Next, we study the tightness of the rank constraint relaxation in $\mathcal{P}_4$.
\begin{thm}
For $P_{s}>0$, if $\mathcal{P}_4$ is feasible, then $\mathrm{Rank}(\boldsymbol{W})\le 1$ is always satisfied.
\end{thm}
\begin{proof}
We adopt a similar approach in \cite{su2018physical} to prove \textbf{Theorem~1}. By relaxing the rank constraint in (\ref{P4_g}), the relaxed problem is jointly convex with respect to the optimization variables, and it satisfies Slater's constraint qualification. Therefore, to reveal the structure of $\boldsymbol{W}$, the Lagrangian function is given by
\begin{equation}
\begin{aligned}
L&=-\xi\mathrm{Tr}(\boldsymbol{WH}) -\mathrm{Tr}\left( \boldsymbol{QW} \right)
+\chi \mathrm{Tr}\left( \boldsymbol{W} \right)
-\sum_{i=1}^{M} \delta _i \mathrm{Tr}\left( \boldsymbol{G}_i\boldsymbol{W} \right)\\
&+\sum_{i=1}^{M}\varphi _i \mathrm{Tr}\left( \boldsymbol{G}_i\boldsymbol{W} \right)
 +\sum_{k=1}^{K} \rho_k \mathrm{Tr}\left({\boldsymbol{H}}_{eve,k} \boldsymbol{W} \right)+\Delta,
\end{aligned}
\end{equation}
where $\boldsymbol{Q} \succeq \boldsymbol{0}$, $\chi$ $\geq$ 0, $\delta_{i}$ $\geq$ 0, $\varphi _{i}$ $\geq$ 0, $\rho_{k}$ $\geq$ 0, $\xi$ $\geq$ 0 are the dual variables for constraints in $\mathcal{P}_4$. $\Delta$ consists of all terms that are not correlated with $\boldsymbol{W}$. Then, we reveal the structure of $\boldsymbol{W}$ by checking the Karush-Kuhn-Tucker (KKT) conditions of problem $\mathcal{P}_4$, which are expressed as
\begin{equation}
\begin{aligned}
&\mathrm{K1}:\chi ^*,\delta _{i}^{*},\varphi _{i}^{*},\rho_{k}^{*}, \xi^*\ge 0, \boldsymbol{Q}^*\succeq \boldsymbol{0},\\
&\mathrm{K2}:\boldsymbol{Q}^*\boldsymbol{W}^*=\boldsymbol{0},~\mathrm{K3}: \nabla _{\boldsymbol{W}^*}L=\boldsymbol{0}.
\end{aligned}
\end{equation}\par
KKT condition K3 can be rewritten as
\begin{equation}
\begin{aligned}
\nabla _{\boldsymbol{W}^*}L&=\chi ^*\boldsymbol{I}-\sum_{i=1}^{M} \left( \delta _{i}^{*}-\varphi _{i}^{*} \right) \boldsymbol{G}_i +\sum_{k=1}^{K}\rho_{k}^{*}{\boldsymbol{H}}_{eve,k}\\
&\quad -\xi\boldsymbol{H}-\boldsymbol{Q}^*=\boldsymbol{0},
\end{aligned}
\end{equation}
resulting in
\begin{equation}
\boldsymbol{Q}^*=\chi ^*\boldsymbol{I}-\sum_{i=1}^{M}\left( \delta _{i}^{*}-\varphi_{i}^{*} \right) \boldsymbol{G}_i+\sum_{k=1}^{K}\rho_{k}^{*}{\boldsymbol{H}}_{eve,k}-\xi\boldsymbol{H}.
\end{equation}
By exploiting \cite[Appendix-proof of Theorem 11.1]{su2018physical}, it can be proved that Rank($\boldsymbol{Q}$)$\geq$ $N_{t}-1$. According to K2, inequality Rank($\boldsymbol{W}$)$\leq 1$ holds, which completes the proof.
\end{proof}
Therefore, the optimal beamforming vector $\boldsymbol{w}^*$ can be obtained by performing eigenvalue decomposition of $\boldsymbol{W}^*$, i.e., $\boldsymbol{W}^*=\boldsymbol{w}^*(\boldsymbol{w}^*)^H$.

\subsection{Reflective Beamforming Design at the IRS}\label{subC}
Next, we optimize the reflective beamforming at the IRS for given $\boldsymbol{w}$ and $\boldsymbol{V}$.
Let $\boldsymbol{a}_1=\mathrm{diag}(\boldsymbol{h}_r^H)\boldsymbol{Q}$, $\boldsymbol{a}_2=\boldsymbol{h}_d$, $\boldsymbol{b}_{r,i}=\mathrm{diag}(\boldsymbol{g}^{H}_{r,i})\boldsymbol{Q}$, $\boldsymbol{b}_{d,i}=\boldsymbol{g}_{d,i}$, $\boldsymbol{d}_{r,k}=\mathrm{diag}(\bm{h}^H_{reve,k})\bm{Q}$, $\boldsymbol{d}_{d,k}=\boldsymbol{h}_{deve,k}$, $\bm{u}=[e^{j \theta_1}, \ldots, e^{j \theta_{N_r}}]^H$, and $\bar{\bm{u}}=[\bm{u}; 1]$.
We can get $|\bm{u}|=1$ and
\begin{equation}
\begin{aligned}
    &\bm{h}^H=\bm{u}^H \bm{a}_1+ \bm{a}^H_2,~\bm{g}^H_{i} =\bm{u}^H \bm{b}_{r,i}+\boldsymbol{b}^H_{d,i},\\
    &\bm{h}^H_{eve,k}=\bm{u}^H \bm{d}_{r,k}+\bm{d}^H_{d,k}.
\end{aligned}
\end{equation}\par
Then, the objective and constraints (\ref{P3_c})-(\ref{P3_e}) can be rewritten as
\begin{subequations}
\begin{align}
&\left|\boldsymbol{h}^{H} \boldsymbol{w}\right|^{2}= \bar{\bm{u}}^H \bm{A}_1 \bar{\bm{u}}+\bm{a}^H_2 \bm{W} \bm{a}_2,\\
&\bar{\bm{u}}^H (\bm{A}_{4,i} + \bm{A}_{5,i} )\bar{\bm{u}}+\bm{b}_{d,i}^H (\bm{W}+\bm{V}) \bm{b}_{d,i} \geq \beta_i,\\
&\bar{\bm{u}}^H (\bm{A}_{4,i}-\gamma \bm{A}_{5,i}) \bar{\bm{u}}+\bm{b}_{d,i}^H \bm Z \bm{b}_{d,i} \leq \gamma \sigma^2_{e,i},\\
&\bar{\bm{u}}^H (\bm{A}_{2,k}-\gamma \bm{A}_{3,k})
\bar{\bm{u}}+\bm{d}^H_{d,k} \bm Z \bm{d}_{d,k} \leq \gamma \sigma^2_{eve,k},
\end{align}
\end{subequations}
where
\begin{equation}
\begin{aligned}
    \bm{A}_1&=\begin{bmatrix}
    \bm{a}_1 \bm W \bm{a}^H_1 & \bm{a}_1 \bm{W} \bm{a}_2\\
    \bm{a}^H_2 \bm{W} \bm{a}^H_1  & 0
    \end{bmatrix},
    \bm{A}_{2,k}=\begin{bmatrix}
    \bm d_{r,k} \bm W \bm d^H_{r,k} & \bm {d}_{r,k} \bm{W} \bm{d}_{d,k}\\
    \bm{d}^H_{d,k} \bm{W} \bm{d}^H_{r,k}  & 0
    \end{bmatrix},\\
    \bm{A}_{3,k}&=\begin{bmatrix}
    \bm d_{r,k} \bm V \bm d^H_{r,k} & \bm {d}_{r,k} \bm{V} \bm {d}_{d,k}\\
    \bm{d}^H_{d,k} \bm{V} \bm{d}^H_{r,k}  & 0
    \end{bmatrix},
    \bm{A}_{4,i}=\begin{bmatrix}
    \bm{b}_{d,i} \bm W \bm b^H_{r,i} & \bm{b}_{r,i} \bm W \bm b_{d,i}\\
    \bm b_{d,i}^H  \bm W \bm b_{r,i}^H  & 0
    \end{bmatrix},\\
    \bm{A}_{5,i}&=\begin{bmatrix}
    \bm b_{r,i} \bm V \bm b^H_{r,i} & \bm b_{r,i} \bm V \bm b_{d,i}\\
    \bm b_{d,i}^H \bm V \bm b_{r,i}^H  & 0
    \end{bmatrix},~ \bm Z=\bm W- \gamma \bm{V}.
\end{aligned}
\end{equation}

We definite $\bm{U}=\bar{\bm{u}} \bar{\bm{u}}^H$. $\bm{U} \in \mathbb{C}^{(N_r+1)\times(N_r+1)}$ is semi-definite and $ \mathrm{rank}(\bm{U})=1$. Set $\bm B_{i}=\bm{b}_{d,i} \bm{b}_{d,i}^H$, $\bm C= \bm{a}_2 \bm{a}^H_2$ and $\bm D_{k}=\bm{d}_{d,k} \bm{d}_{d,k}^H$. Then, problem $\mathcal{P}_3$ can be converted to the following problem.
\begin{subequations}
\begin{align}
\mathcal{P}_5:
&\min_{{\bm U}}~ -\log_2(1+(\mathrm{Tr}(\bm A_1 \bm U)+\mathrm{Tr}(\bm C \bm W))/\sigma_{ir}^2) \label{P5_a}\\
\mbox{s.t.}~
&\mathrm{Tr} \left((\bm{A}_{4,i} + \bm{A}_{5,i} ) \bm U \right) + \mathrm{Tr} (\bm{B}_{i} (\bm W+ \bm V))\geq \beta_{i},&\label{P5_b}\\
&\mathrm{Tr} \left((\bm{A}_{4,i}-\gamma \bm{A}_{5,i} ) \bm U \right) + \mathrm{Tr} (\bm{B}_{i} \bm Z) \leq \gamma \sigma^2_{e,i}, &\label{P5_c}\\
&\mathrm{Tr} \left((\bm{A}_{2,k}-\gamma \bm{A}_{3,k} ) \bm U \right) + \mathrm{Tr} (\bm{D}_{k} \bm Z) \leq \gamma \sigma^2_{eve,k}, &\label{P5_d}\\
& \bm{U}_{n,n}=1,~\forall n \in \{1,2,\cdots, N_r+1\}, ~\bm U \succeq \boldsymbol{0},&\label{P5_e}\\
& \mathrm{Rank}(\bm U)=1.&\label{P5_f}
\end{align}
\end{subequations}

Because of the rank-one constraint, $\mathcal{P}_5$ is non-convex. To tackle this problem, we apply a penalty-based algorithm \cite{qi2020robust,mao2019power,nocedal2006numerical}. Since $\bm{U}$ is a PSD matrix, $\mathrm{Tr}(\bm U) \geq \lambda_{max}(\bm U)$ holds. Moreover, $\mathrm{Rank}(\bm U)=1$ exists when its trace is equal to its maximum eigenvalue, i.e., $\mathrm{Tr}(\bm U) = \lambda_{max}(\bm U)$. To overcome the non-convexity, we replace rank-one constraint with a penalty function $\eta(\mathrm{Tr}(\bm U) - \lambda_{max}(\bm U))$. The penalty function is moved into the objective function, which results in the following optimization problem.
\begin{equation}
\begin{aligned}
\mathcal{P}_6:
\min_{{\bm U}}&~-\log_2(1+(\mathrm{Tr}(\bm A_1 \bm U)+\mathrm{Tr}(\bm C \bm W))/\sigma_{ir}^2)\\
&+\eta(\mathrm{Tr}(\bm U) - \lambda_{max}(\bm U)) \label{P6_a}\\
\mbox{s.t.}~
&(\text{\ref{P5_b}})-(\text{\ref{P5_e}}),
\end{aligned}
\end{equation}
where $\eta >0$ is a penalty factor. It is obvious that the rank-one solution of $\bm U$ can be obtained when $\eta$ is sufficiently large. \par
Due to the convexity of $\lambda_{max}(\bm U)$, the problem $\mathcal{P}_6$ is still non-convex. To solve this problem, we employ an effective successive convex approximation (SCA) method to transform it into the following iterative optimization problem.
\begin{subequations}
\begin{align}
\mathcal{P}_7:
\min_{{\bm U}}&~-\log_2(1+(\mathrm{Tr}(\bm A_1 \bm U^{(t+1)}+\mathrm{Tr}(\bm C \bm W)))/\sigma_{ir}^2) \nonumber\\
&~+\eta (\mathrm{Tr}(\bm U^{(t+1)})-(\bar{\bm u}^{(t)}_{max})^H \bm U^{(t+1)} \bar{\bm u}^{(t)}_{max}) \label{P7_a}\\
\mbox{s.t.}~~
&(\text{\ref{P5_b}})-(\text{\ref{P5_e}}),
\end{align}
\end{subequations}
where superscript $(t)$ represents the iteration index of the optimization variables.\par
Then, $\mathcal{P}_7$ can be solved by CVX. We can obtain $\bm U= \lambda_{max}(\bm U) \bar{\bm u}_{max} \bar{\bm u}^H_{max}$ when $\mathrm{Tr}(\bm U) \approx \lambda_{max}(\bm U)$. $\bar{\bm u}_{max}$ denotes the unit eigenvector related to the maximum eigenvalue $\lambda_{max}(\bm U)$. Then, we can obtain the optimal reflecting vector $\bar{\bm u}=\sqrt{\lambda_{max}(\bm U)} \bar{\bm u}_{max}$. Finally, the optimal solution $\bm{u}^{*}$ can be expressed as $\bm{u}^{*}=[\bar{\bm{u}}/\bar{\bm{u}}_{N_r+1}]_{(1:N_r)}$.\par
The procedure of the penalty-based algorithm is summarized in \textbf{Algorithm \ref{algo-1}}. The convergence of the proposed algorithm can be proved in \cite[Theorem 1]{mao2019power}. Then, the proposed overall AO-based algorithm is summarized in \textbf{Algorithm \ref{algo-2}}.

\begin{algorithm}[t]
\caption{Penalty-based Algorithm for Obtaining Reflective Beamforming.} 
\label{algo-1}
\hspace*{0.02in}{\bf Initialization:} Set feasible values $\bm{U}^{(0)}$, penalty factor $\eta$, and convergence tolerance $\epsilon$.\\
\hspace*{0.02in}{\bf For} $t=0,1,2,\cdots$ \textbf{do}\\
Solve problem $\mathcal{P}_7$ by using CVX to obtain $\bm U^{(t+1)}$.\\
\hspace*{0.02in}{\bf If} $\mathrm{Tr}(\bm U^{(t+1)})-\lambda_{max}(\bm U^{(t+1)}) \leq \epsilon$\\
~~~~~\textbf{break}\\
\hspace*{0.02in}{\bf End}\\
\hspace*{0.02in}{\bf End}\\
\hspace*{0.02in}{\bf Output:} 
Optimal reflective beamforming $\bm u^{*}$.\\
\end{algorithm}
\begin{algorithm}[t]
\caption{AO-based Algorithm for Solving $\mathcal{P}_3$.}
\label{algo-2}
\hspace*{0.02in}{\bf Initialization:} Set $\bm{\Theta}^{(0)}$ and iteration index $r=0$.\\
\hspace*{0.02in}{\bf Repeat:}\\
For given $\bm{\Theta}^{(r)}$, solve problem $\mathcal{P}_4$ by using CVX to obtain the solution $\bm{w}^{(r+1)}$ and $\bm{V}^{(r+1)}$.\\
For given $\bm{w}^{(r+1)}$ and $\bm{V}^{(r+1)}$, solve problem $\mathcal{P}_7$ by \textbf{Algorithm} \textbf{\ref{algo-1}} to obtain $\bm{\Theta}^{(r+1)}=\mathrm{diag}(\bm u^{(t+1)})$.\\
Update $r=r+1$.\\
\hspace*{0.02in}{\bf Until:}
Convergence\\
\end{algorithm}

\section{Simulation Results}\label{IV}
In this section, we provide some numerical results to validate the effectiveness of the proposed scheme in the secure IRS-aided SWIPT system. We consider a two-dimensional coordinate system. The BS is located at $(0, 0)$ meters, and the IRS is located at $(5,3)$ meters. Two ERs and two Eves are randomly deployed within a circular area centered at $(5,0)$ meters with radius $1$ m and $(55,0)$ meters with radius $2$ m, respectively. One IR is deployed at $(50,0)$ meters. The path loss is composed of large-scale path loss (e.g. distance-dependent path loss \cite{wu2019beamforming}) and small-scale path loss (e.g., the considered channels follows a Rayleigh distribution). The path loss exponent $\alpha$ for BS-ER link, BS-IR/Eves link, IRS-IR/ERs/Eves link, and BS-IRS link are set as $3$, $3.6$, $2.5$, $2$, respectively. Unless otherwise stated, other parameters are given as: $M_i=24$ mW, $a_i=150$, $b_i=0.014$ \cite{qi2020robust}, $N_{t}=16$, $\sigma^2_{ir}=\sigma^2_{e,i}=\sigma^2_{eve,k}=-60$ dBm, $P_s=30$ dBm, and $\mu_i=10$ $\mu$W. Next, we employ two baselines for comparison. For baseline 1, the system performance without IRS is evaluated. For baseline 2, we employ an IRS with a random phase \cite{di2019smart}.\par
\begin{figure}[!t]
\centering
\includegraphics[scale=0.6]{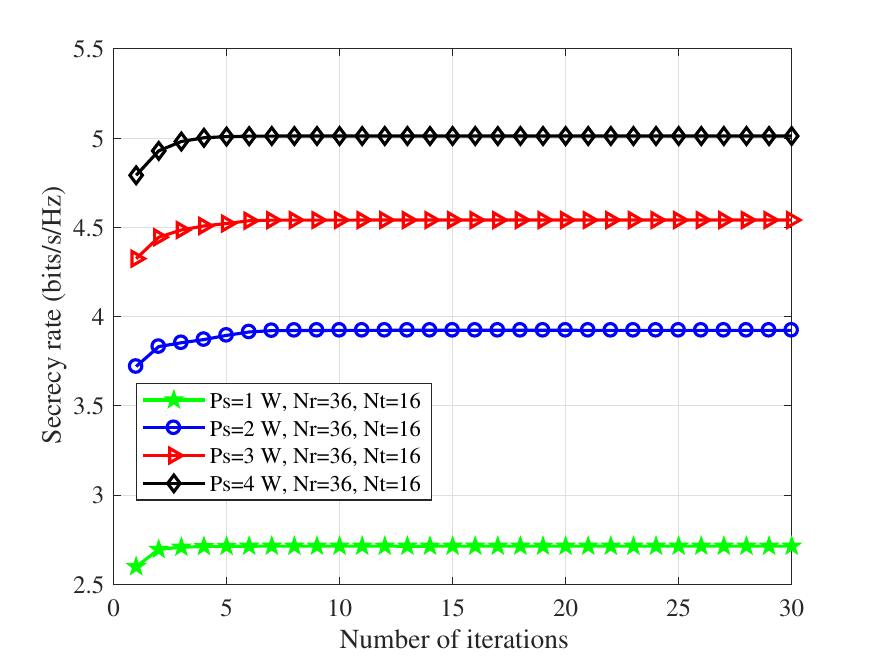}
\caption{Secrecy rate versus number of iterations.\vspace{-12pt}}
\label{fig:2}
\end{figure}
The convergence of the secrecy rate by using the proposed scheme is investigated under different transmit powers shown in Fig. \ref{fig:2}. Each secrecy rate obtained by the proposed scheme converges after around $7$ iterations on average.

Figure \ref{fig:3} shows the secrecy rate versus the total transmit power at the BS, $P_s$, under different numbers of reflecting elements, while setting $\tau=1$ bit/s/Hz. It is observed that the secrecy rates for the proposed scheme and baseline schemes increase monotonically with the transmit power. The proposed scheme outperforms the baseline schemes. This is because the proposed scheme can provide a more favorable wireless propagation environment for the IR and weaken the interests of eavesdroppers. In addition, increasing the number of reflecting elements improves the secrecy rate. The reason is that degrees of freedom (DoFs) increase with the number of reflecting elements, so as to improve the secrecy rate of the system. Specifically, when $N_r=64$, our proposed scheme can improve the secrecy rate performance by 80.4\% and 27.3\% on average compared with baseline 1 and baseline 2, respectively, and when $N_r=36$, it can improve 44.5\% and 24.9\%, respectively.

\begin{figure}[!t]
\centering
\includegraphics[scale=0.55]{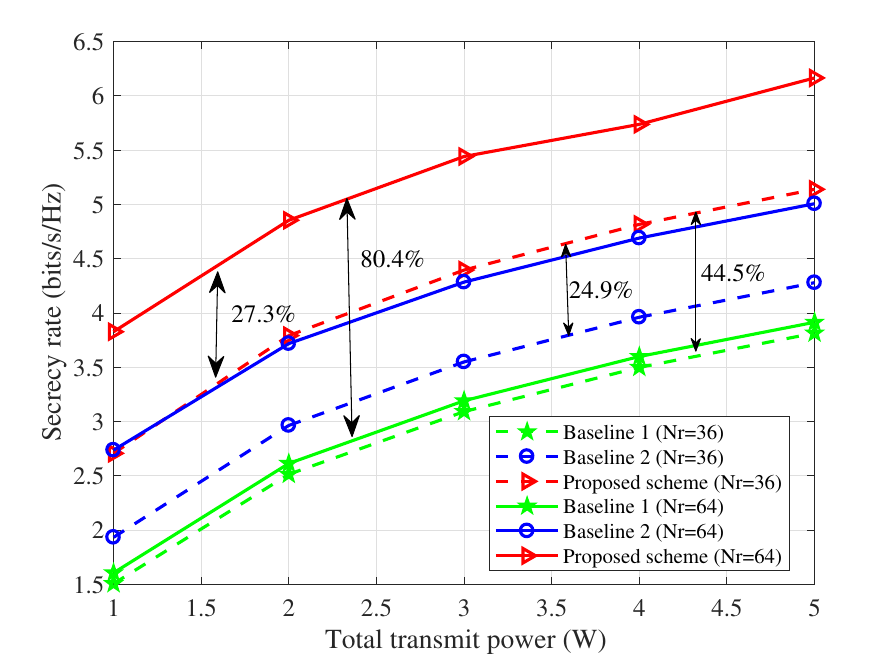}
\caption{Secrecy rate versus total transmit power.\vspace{-12pt}}
\label{fig:3}
\end{figure}

Figure \ref{fig:4} depicts the secrecy rate versus the harvested power requirement when $P_s=1$ W. Evidently, the secrecy rate decreases as the minimum harvested power requirement increases. The secrecy rate obtained by using the proposed scheme is much higher than that obtained by using other baselines shown in Fig. \ref{fig:4}. Moreover, we also study two scenarios with the different numbers of the transmit antenna $N_{t}=\{8, 16\}$. As can be observed, the secrecy rate of the system increases with $N_t$.

\begin{figure}[!t]
\centering
\includegraphics[scale=0.58]{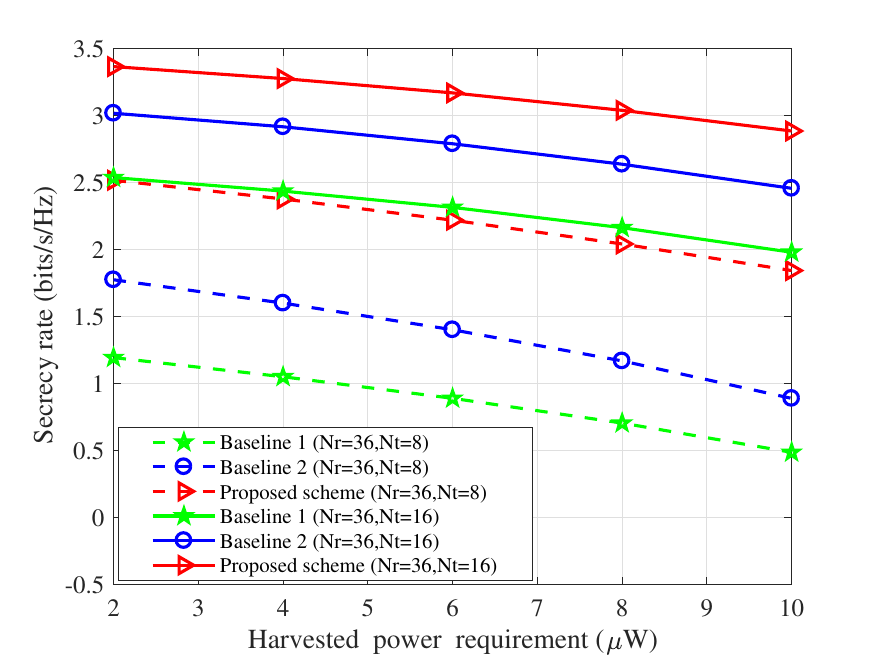}
\caption{Secrecy rate versus harvested power requirement. \vspace{-12pt}}
\label{fig:4}
\end{figure}

\section{Conclusion}\label{V}
In this paper, we have presented an investigation of secure beamforming design in the secure IRS-aided SWIPT systems to maximize the secrecy rate while satisfying the BS transmit power and EH constraints. An AO algorithm was proposed to tackle the coupling of optimization variables. Next, the transmit beamforming and AN covariance matrix were optimized by adopting the SDR method. The penalty-based algorithm was employed to optimize the reflective beamforming. Simulation results show that our proposed scheme performs better than the baseline schemes in terms of secrecy rate.

\ifCLASSOPTIONcaptionsoff
  \newpage
\fi

\bibliographystyle{ieeetr}

\end{document}